\documentclass{article}
\usepackage{amsmath}
\usepackage{amsthm}
\usepackage{indentfirst}
\usepackage{subfigure}
\usepackage{graphicx}
\usepackage{epsfig}
\usepackage[body={16cm,24cm}, top=3cm]{geometry}
\newtheorem{definition}{Definition}
\newtheorem{theorem}{Theorem}
\newtheorem{proposition}{Proposition}
\newtheorem{corollary}{Corollary}
\newtheorem{example}{Example}
\newtheorem{lemma}{Lemma}
\usepackage{amssymb}
\usepackage{indentfirst}
\author{Wang Changlong, Jigen Peng}
\title{Analysis of equivalence relation in joint sparse recovery}
\begin{document}\large
\title{Analysis of the equivalence relationship in joint sparse recovery}
\date{}
\maketitle
\begin{abstract}
The joint sparse recovery problem is a generalization of the single measurement vector problem which is widely studied in Compressed Sensing and it aims to recovery a set of jointly sparse vectors. i.e. have nonzero entries concentrated at common location. Meanwhile $l_p$-minimization subject to matrices is widely used in a large number of algorithms designed for this problem. Therefore the main contribution in this paper is two theoretical results about this technique. The first one is to prove that in every multiple systems of linear equation, there exists a constant $p^{\ast}$ such that the original unique sparse solution also can be recovered from a minimization in $l_p$ quasi-norm subject to matrices whenever $0< p<p^{\ast}$. The other one is to show an analysis expression of such $p^{\ast}$. Finally, we display the results of one example to confirm the validity of our conclusions.
\end{abstract}
%\begin{flushleft}
%\textbf{keywords闁挎稒鐣?sparse recovery ;\quad null space constant;\quad null space property
%\end{flushleft}
\textbf{keywords:} sparse recovery, multiple measurement vectors, joint sparse recovery, null space property, $l_p$-minimization
%\begin{keywords}
%sparse recovery; null space constant; null space property; $l_0$-minimization; $l_p$-minimizationh
%\end{keywords}
%section1
\section{INTRODUCTION}
In sparse information processing, one of the central problems is to recovery a sparse solution of an underdetermined linear system, such as visual coding \cite{olshausen1996emergence}, matrix completion \cite{candes2009exact}, source localization \cite{malioutov2005sparse}, and face recognition \cite{wright2009robust}. That is, letting $A$ be an underdetermined matrix of size $m \times n$ and $b \in \mathbb{R}^m$ is a vector representing some signal, so the single measurement vector (SMV) is popularly modeled into the following $l_0$-minimization.
%One of the core problems in Compressed Sensing is to find the sparest solution to the underdetermined system $Ax=b$, where $A \in R^{n \times m}$ is an underdetermined matrix (i.e. $n<m$), and $b \in R^m$ a vector representing some measured or represented signals. The problem is popularly modeled into the following $l_0$-minimization:
\begin{eqnarray}
\mathop{\min}\limits_{x \in \mathbb R^m} \|x\|_0\ s.t. \ Ax=b,
\end{eqnarray}
where $\|x\|_0$ indicates the number of nonzero elements of $x$. However, $l_0$-minimization has been proved to be NP-hard \cite{natarajan1995sparse} because of the discrete and discontinuous nature of $\|x\|_0$ . In order to overcome this difficulty, many researchers have suggested to replace $\|x\|_0$ with $\|x\|_p^p$. Instead of $l_0$-minimization, they consider the $l_p$-minimization with $0\leq p \leq1$,
\begin{eqnarray}
\mathop{\min}\limits_{x \in \mathbb R^m} \|x\|_p^p \ s.t. \ Ax=b
\end{eqnarray}
where $\|x\|_p^p=\sum_{i=1}^m |x_i|^p$ (\cite{foucart2009sparsest} \cite{chartrand2007exact}). Due to the fact that $\|x\|_0=\mathop{lim}\limits_{p \to 0} \|x\|_p^p$, it seems to be more natural to consider $l_p$-minimization.

Furthermore, a natural extension of single measurement vector is the joint sparse recovery problem, also known as the multiple measurement vector (MMV) problem which arises naturally in source localization \cite{Malioutov2003A}, neuromagnetic imaging \cite{Cotter2005Sparse}, and equalization of sparse-communication channels \cite{Fevrier1999Reduced} \cite{Cotter2002Sparse}. Instead of a single measurement $b$, we are given a set of $r$ measurements,
\begin{eqnarray}
Ax^{(k)}=b^{(k)} \qquad k=1 \dots r,
\end{eqnarray}
in which the vectors $x^{(k)} \ ( k=1 \dots r)$ are joint sparse, i.e.  the solution vectors share a common support and have nonzero entries concentrated at common locations.

Let $A \in \mathbb R^{m \times n}$ and $B=[b^{(1)} \dots b^{(r)}] \in \mathbb{R}^{m \times r}$ , the MMV problem is to look for the row-sparse solution matrix and it can be modeled as the following $l_{2,0}$-minimization problem.
\begin{eqnarray}
\mathop{\min}\limits_{X \in \mathbb{R}^{m\times r}} \|X\|_{2,0}\ s.t. \ AX=B,
\end{eqnarray}
where $\|X\|_{2,0}=\sum_{i=1}^n \|X_{row \ i}\|_{2,0}$ and $X_{row \ i}$ is a row vector and defined as the $i$-th row of $X$, and $\|X_{row \ i}\|_{2,0}=1$ if $\|X_{row \ i}\|_2 \neq 0$ and $\|X_{row \ i}\|_{2,0}=0$ if $\|X_{row \ i}\|_2 =0$.

We can define the support of $X$, $support(X)=S=\{i:\ \|X_{row \ i}\|_2 \neq 0\}$ and call the solution $X$ is $k$-sparse, when $|S| \leq k$, where $|S|$ is the cardinality of set $S$ and we also say that $X$ can be recovered by model (4), if $X$ is the unique solution to model (4).

It needs to be emphasized that we can not regard the solution of multiple measurement vector (MMV) as a combination of several solutions of single measurement vectors. i.e., the solution matrix $X$ to $l_{2,0}$-minimization is not always composed by the solution vectors to $l_0$-minimization. For example.

\begin{example}
We consider an underdetermined system $AX=B$, where
\begin{center}
$A=\left(
\begin{array}{ccccc}
2 & 0 & 0 & 1 & 0 \\
0 & 0.5 & 0 & 1 & 0 \\
0 & 0 & 1 & 2 & -0.5 \\
0 & 0 & 0 & -1 & 0.5
\end{array}
\right)$ and $B=\left(
\begin{array}{cc}
1 & 1  \\
1 & 1  \\
0 & 1  \\
0 & 0
\end{array}
\right).$
\end{center}
\end{example}

If we treat the $AX=B=[b_1 \ b_2 ]$ as a combination of two single measurements vector, $Ax=b_1$ and $Ax=b_2$, it is easy to verify that each sparse solution to these two problems is $x_1=[0.5\ 2\ 0\ 0]^T$,and $x_2=[0\ 0\ 1\ 2]^T$ . So let $X^*=[x_1 \ x_2]$, it is easy to check that $\|X^*\|_{2,0}=4$.
In fact, it is easy to verify that
\begin{center}
 $X=\left(
\begin{array}{ccc}
0.5 & 0.5 \\
2 & 2 \\
0 & 1 \\
0 & 0 \\
\end{array}
\right)$
\end{center}
$X$ is the solution to $l_{2,0}$-minimization since $\|X\|_{2,0}=3 < \|X^*\|_{2,0}=4 $.

With this simple Example 1, we should be aware that MMV problem wants a jointly sparse solution, not a solution which is just composed by sparse vectors. Therefore, MMV problem is more complex than SMV, so MMV needs its own theoretical work. Be inspired by $l_p$-minimization, a popular approach to find the sparest solution to MMV problem is to solve the following $l_{2,p}$-minimization optimization problem.
\begin{eqnarray}
\mathop{\min}\limits_{X \in \mathbb R^(m\times r)} \|X\|_{2,p}\ s.t. \ AX=B,
\end{eqnarray}
where the mixed norm $\|X\|_{2,p}^p=\sum_{i=1}^n\|X_{row \ i }\|_2^p$ and $p\in (0,1].$

\subsection{Related Work}

Many researchers have made a lot of contribution related to the existence, uniqueness and other properties of $l_{2,p}$-minimization \cite{Liao2015Analysis}\cite{Foucart2010Real}\cite{Lai2011The}\cite{Van2010Theoretical}.
Eldar \cite{Eldar2009Beyond} gives a sufficient condition for MMV when $p=1$, and Unser \cite{Unser2000Sampling} analyses some properties of the solution to $l_{2,p}$-minimization when $p=1$.
Fourcart and Gribonval \cite{Foucart2010Real} studied the MMV setting when $r=2$ and $p=1$, they gave a sufficient and necessary condition to judge whether a $k$-sparse matrix $X$ can be recovered by $l_{2,p}$-minimization. Furthermore, Lai and Liu \cite{Lai2011The} consider the MMV setting when $r \geq 2$ and $p\in [0,1]$, they improved the condition in \cite{Foucart2010Real} and give a sufficient and necessary condition when $r\geq2$ .

On the other hand, numerous algorithms have been proposed and studied for $l_{2,0}$-minimization (e.g. \cite{Hyder2010Direction} \cite{Hyder2009A}). Orthogonal Matching Pursuit (OMP) algorithms are extended to the MMV problem \cite{Tropp2007Signal}, and convex optimization formulations with mixed norm extend to the corresponding the SMV solution \cite{Milzarek2014A}.
Hyder \cite{Hyder2009A} provides us a robust algorithm for $l_{2,p}$-minimization which shows a clear improvement in both noiseless and noisy environment.

Due to the fact that $\|X\|_{2,0}=\mathop{lim}\limits_{p \to 0} \|X\|_{2,p}^p$, it seems to be more natural to consider $l_{2,p}$-minimization instead of a NP-hard optimization $l_{2,0}$-minimization than others. However, it is an important theoretical problem that whether there exists a general equivalence relationship between $l_{2,p}$-minimization and $l_{2,0}$-minimization.

In the case $r=1$, Peng \cite{peng2015np} have given a definite answer to this theoretical problem. There exists a constant $p(A,b)>0$, such that every a solution to $l_p$-minimization is also the solution to $l_0$-minimization whenever $0<p<p(A,b)$,
\begin{eqnarray}
p(A,b)=\frac{\ln\left(\mathop{\min}\limits_{Ax=b}\|x\|_0+1 \right)-\ln\left(\mathop{\min}\limits_{Ax=b}\|x\|_0 \right)}{\ln r-\ln r_m},
\end{eqnarray}
However, this range can not be calculated.

Peng \cite{peng2015np} only proves the conclusion when $r=1$, so it is urgent to extend this conclusion to MMV problem. Furthermore, Peng just proves the existence of such $p$, he does not give us a computable expression of such $p$. Therefore, the main purpose of this paper is not only to prove the equivalence relationship between $l_{2,p}$-minimization and $l_{2,0}$-minimization, but also present an analysis expression of such $p$ in Section 2 and Section 3.

\subsection{Main Contribution}
In this paper, we focus on the equivalence relationship between $l_{2,p}$-minimization and $l_{2,0}$-minimization. Furthermore, it is an application problem that an analysis expression of such $p^{*}$ is needed, especially in designing some algorithms for $l_{2,p}$-minimization.

In brief, this paper gives answers to two problems which are urgently needed to be solved:

(\uppercase\expandafter{\romannumeral1}). There exists a constant $p^{*}$ such that every $k$-sparse solution matrix $X$ to $l_{2,0}$-minimization is also the solution to $l_{2,p}$-minimization whenever $0<p<p^{*}$.

(\uppercase\expandafter{\romannumeral2}). We give an analysis expression of such $p^{*}$ which is formulated by the dimension of the matrix $A\in\mathbb{R}^{m\times n}$, the eigenvalue of the matrix $A^TA$ and $B\in\mathbb{R}^{m\times r}$.

Our paper is organized as follows. In Section 2, we will present some preliminaries of the null space condition, which plays a core role in the proof of our main theorem, and prove the equivalence relationship between $l_{2,p}$-minimization and $l_{2,0}$-minimization. In Section 3 we focus on proving the another main results of this paper. There we will present an analysis expression of such $p^{*}$ . Finally, we summarize our finding in last section.
\subsection{Notation}
For convenience, for $x \in \mathbb R^n$, we define its support by $support\ (x)=\{i:x_i \neq 0\}$ and the cardinality of set S by $|S|$.
Let $Ker(A)=\{x \in \mathbb R^n:Ax=0\}$ be the null space of matrix A, denote by $\lambda_{min}^{+}(A)$ the minimum nonzero absolute-value eigenvalue of $A^TA$ and by $\lambda_{max}(A)$ the maximum one. We also use the subscript notation $x_S$ to denote such a vector that is equal to $x$ on the index set $S$ and zero everywhere else. and use the subscript notation $X_S$ to denote a matrix whose rows are those of the rows of $X$ that are in the set index S and zero everywhere else. Let $X_{col\ i}$ be the $i$-th column in $X$, and let $X_{row\ i}$ be the $i$-th row in $X$. i.e $X=[X_{col\ 1},X_{col\ 2}\dots X_{col\ r}]=[X_{row\ 1},X_{row\ 2}\dots X_{row\ m}]^T$, for $X\in \mathbb{R}^{n \times r}$. We use $\langle A,B\rangle=tr(A^TB)$ and $\|A\|_F=\sum_{i,j}|a_{ij}|^2$.
\section{EQUIVALENCE RELATIONSHIP BETWEEN\\ $l_{2,p}$-MINIMIZATION AND $l_{2,0}$-MINIMIZATION}

In the single measurement vector (SVM) problem, there exists a sufficient and necessary condition to judge a $k$-sparse vector whether can be recovered by $l_0$-minimization and $l_p$-minimization, namely, the null space condition.
\begin{theorem}{\rm \cite{gribonval2003sparse}}
Given a matrix $A \in R^{m \times n}$ with $m \leq n$, every $x^{\ast}$ with $\|x^{\ast}\|_0=k$  can be recovered by $l_p$-minimization $(0\leq p \leq 1)$  if and only if:
\begin{eqnarray}
\|x_S\|_p < \|x_{S^C}\|_p,
\end{eqnarray}
 for any $x \in Ker(A)$, and set $S \subset \{1,2,3 \dots n\}$ with $|S| \leq |T^{\ast}|$, where $T^{\ast}=support(x^{\ast}).$
\end{theorem}

Null space condition is widely used in sparse theory, however, this condition only considers a single measurement which we can treat it as the situation that $r=1$ in MMV problem. Furthermore, in \cite{Lai2011The}, the well-known Null space condition has been extended to the situation when $r >1$.
\begin{theorem}{\rm (Theorem 1.3 of \cite{Lai2011The})}
Let A be a real matrix of size $m \times n$ and $S \subseteq \{1,2 \dots n\}$ be a fixed index set. Fixed $p \in [0,1]$ and $r \geq 1$. Then the following condition are equivalent

(a) All $x^{(k)}$ with support in $S$ for $k=1 \dots r$ can be uniquely recovered  by $l_{2,p}$-minimization.

(b) For all vectors $Z=[z^{(1)},z^{(2)} \dots z^{(r)}] \in (N(A))^r \backslash \{(\textbf 0.\textbf 0 \dots \textbf 0)\}$
\begin{eqnarray}
\|Z_S\|_{2,p} < \|Z_{S^C}\|_{2,p}.
\end{eqnarray}

(c) For all vectors $z\in N(A)$, we have $\sum_{j \in S} |z_j|^p< \sum_{j \in S^C} |z_j|^p$.
\end{theorem}

It is worth pointing out that Theorem 2 not only provides us a sufficient and necessary condition of MMV's version, but also proves the equivalence relationship between the situations when $r=1$ and $r>1$.

According to Theorem 2, we can get the following corollary which is very easy to be proved.
\begin{corollary}
Given a matrix $A\in \mathbb{R}^{m\times n}$, if every $X^{*}\in \mathbb{R}^{m\times r}$ with $\|X^{*}\|_{2,0}=k$ can be recovered by $l_{2,0}$-minimization, then we have the following conclusion.

(a) For any $X\in (N(A))^r \backslash \{(\textbf 0,\textbf 0 \dots \textbf 0)\}$, we have that $\|X\|_{2,0} \geq 2k+1$.

(b) we have that $\displaystyle k \leq \lceil \frac{n-2.5}{2} \rceil+1$, where $\lceil a \rceil$ represents the integer part of $a$.

(c) The number of measurements needed to recovery every $k$-sparse matrices always satisfies $m\geq 2k$, furthermore, $\displaystyle k \leq \lceil \frac{m}{2} \rceil$.
\end{corollary}
\begin{proof}

(a) According to Theorem 2, for any $X\in (N(A))^r \backslash \{(\textbf 0,\textbf 0 \dots \textbf 0)\}$ and $|S|\leq k$, we have that
 \begin{eqnarray}
 \|X_S\|_{2,0} < \|X_{S^C}\|_{2,0},
 \end{eqnarray}
  and it is easy to get that
  \begin{eqnarray}
   \|X\|_{2,0} \geq 2k+1.
\end{eqnarray}
(b) According to the proof of (a), we have that $n \geq \|X\|_{2,0}\geq 2k+1$. Due to the integer-value of $k$, we have that $k\leq (n-1)/2 $ when $n$ is an odd number, similarly, we get that $k\leq (n-2)/2$ when $n$ is an even number.

In brief, we get that $\displaystyle k \leq \lceil \frac{n-2.5}{2} \rceil+1$, where $\lceil a \rceil$ represents the integer part of $a$.

(c) For any $\tilde{x}\in N(A)\setminus \{\textbf 0\}$, we consider $\tilde{X}=[\tilde{x},\tilde{x}\ldots \tilde{x}]\in (N(A))^r$.

According to the proof of (a), it is obvious that $\|\tilde{x}\|_0=\|\tilde{X}\|_{2,0}\geq 2k+1$, such that the sub-matrix $A_S$ is an invertible matrix, where $S=support(\tilde{x})$. Therefore, we can get that $2k\leq rank(A)\leq m<n$. Due to the integer-value virtue of $k$, we also can say that $k \leq \lceil \frac{m}{2} \rceil$.
\end{proof}
In order to clear further the meaning of new version null space condition and use it more conveniently, it is necessary to introduce a new concept named M-null space constant (M-NSC).
\begin{definition}
Given an underdetermined matrix $A\in \mathbb{R}^{m\times n},$ for every $p\in [0,1]$ and a positive integer $k$, the M-null space constant $h(p,A,r,k)$ is the smallest number such that,
$$\|X_S\|_{2,p}^p \leq h(p,A,r,k) \|X_{S^C}\|_{2,p}^p, \ when \ 0<p\leq1,$$
and
$$\|X_S\|_{2,0} \leq h(0,A,r,k) \|X_{S^C}\|_{2,0}, \ when \ p=0,$$
for every index set $S \subset \{1,2,\ldots,n\}$ with $|S| \leq k$ and every $X \in (Ker(A))^r\backslash \{(\textbf 0,\textbf 0\dots \textbf 0)\}.$
\end{definition}
%which is the most important advantage of NSP. However, NSP is difficult to be checked for a given matrix. In order to reach our goal, we recall the concept null space constant (NSC), which is closely related to NSP and will offer tremendous help in illustrating the performance of $l_0$-minimization (1) and $l_p$-minimization (4).

According to the definition of M-NSC, it is easy to get the following corollary which is also very easy to be proved and we leave the proof to readers.
\begin{corollary}
 Every $k$-sparse matrix $X\in \mathbb R ^{n\times r}$ can be recovered by $l_{2,p}$-minimization if and only if $h(p,A,r,k)<1$.
\end{corollary}

As shown in Corollary 2, M-NSC provides us a sufficient and necessary condition of the solution to $l_{2,0}$-minimization and $l_{2,p}$-minimization, and it is important for proofing the equivalence relationship between $l_{2,0}$-minimization and $l_{2,p}$-minimization. Furthermore, we emphasize a few important properties of $h(p,A,r,k)$.
\begin{proposition}
The M-NSC as defined in Definition 1, is nondecreasing in $p \in [0,1].$
\end{proposition}
\begin{proof}
The proof is divided into two steps.

Step 1: To prove $h(p,A,r,k) \leq h(1,A,r,k)$, for any $p\in [0,1].$

For any $X \in (N(A))^r \backslash \{(\textbf 0,\textbf 0 \dots \textbf 0)\}$, without of generality, we assume that $\|X_{row \ 1}\|_2 \geq \|X_{row \ 2}\|_2 \geq \dots \|X_{row \ n}\|_2 $.

We define a function $\theta(p,X,k)$ as
\begin{eqnarray}
\theta(p,X,k)= \frac{\sum_{i=1}^k \|X_{row \ i}\|_2^p}{\sum_{i=k+1}^n\|X_{row \ i}\|_2^p},
\end{eqnarray}
 then it is easy to get that the definition of $h(p,A,r,k)$ is equivalent to
\begin{eqnarray}
h(p,A,r,k)=\mathop{\max}\limits_{|S| \leq k} \mathop{\sup}\limits_{X \in (N(A))^r \backslash \{(\textbf 0,\textbf 0 \dots \textbf 0 )\}} \ \theta (p,X,k)
\end{eqnarray}

For any $p \in [0,1]$, the function $f(t)=\frac {t^p}{t} \ (t>0)$ is a non-increasing function. For any $j \in \{k+1,\dots n \}$ and $i \in \{1,2\dots k\}$, we have that,
\begin{eqnarray}\label{T1}
\frac{\|X_{row \ j}\|_2^p}{\|X_{row \ j}\|_2} \geq \frac {\|X_{row \ i}\|_2^p}{\|X_{row \ i}\|_2}.
\end{eqnarray}

We can rewrite inequalities (\ref{T1}) into
\begin{eqnarray}
\frac{\|X_{row \ i}\|_2^p}{\|X_{row \ j}\|_2^p} \leq \frac {\|X_{row \ i}\|_2}{\|X_{row \ j}\|_2}.
\end{eqnarray}

Therefore, we can get that
\begin{eqnarray}
\frac{\sum_{i=1}^k\|X_{row \ i}\|_2^p}{\|X_{row \ j}\|_2^p} \leq \frac {\sum_{i=1}^k\|X_{row \ i}\|_2}{\|X_{row \ j}\|_2}.
\end{eqnarray}

We can conclude that
\begin{eqnarray}
\frac{\sum_{j=k+1}^n\|X_{row \ j}\|_2^p}{\sum_{i=1}^k\|X_{row \ i}\|_2^p} \geq \frac {\sum_{j=k+1}^n\|X_{row \ j}\|_2}{\sum_{i=1}^k\|X_{row \ i}\|_2}.
\end{eqnarray}
such that $\frac {1}{\theta (p,X,k)} \geq\frac{1}{\theta (1,X,k)}$. i.e., $\theta (p,X,k) \leq \theta (1,X,k)$.

Because $h(p,A,r,k)=\mathop{\max}\limits_{|S| \leq k} \mathop{\sup}\limits_{X \in (N(A))^r \backslash \{(\textbf 0,\textbf 0 \dots )\}} \ \theta (p,X,k)$, we can get that $h(p,A,r,k) \leq h(1,A,r,k)$.

Step 2: To prove $h(pq,A,r,k) \leq h(p,A,r,k)$ for any $p\in [0,1]$ and $q\in(0,1).$

According to the definition of $\theta(p,X,k)$ in Step 1, we have that
\begin{eqnarray}
\theta(pq,X,k)=\frac{\sum_{i=1}^k\|X_{row \ i}\|_2^{pq}}{\sum_{j=k+1}^n\|X_{row \ j}\|_2^{pq}}=\frac{\sum_{i=1}^k(\|X_{row \ i}\|_2^{p})^q}{\sum_{j=k+1}^n(\|X_{row \ j}\|_2^{p})^q} \leq \frac {\sum_{i=1}^n\|X_{row \ i}\|_2^p}{\sum_{j=k+1}^n\|X_{row \ j}\|_2^p}.
\end{eqnarray}

It needs to be pointed out that we have prove the fact in Step 1, that
\begin{eqnarray}
\frac{\sum_{i=1}^k|u_i|^p}{\sum_{j=k+1}^n|u_j|^p} \leq \frac{\sum_{i=1}^k|u_i|}{\sum_{j=k+1}^n|u_j|},
\end{eqnarray}
for any $|u_1|\geq|u_2|\dots \geq |u_n|.$

Therefore, we can get that $\theta(pq,X,k)\leq \theta(p,X,k)$, in other words, $\theta(p_1,X,k)\leq \theta(p_2,X,k)$ as long as $p_1 \leq p_2$.

Because $h(p,A,r,k)=\mathop{\max}\limits_{|S| \leq k} \mathop{\sup}\limits_{X \in (N(A))^r \backslash \{0,0 \dots 0\}} \ \theta (p,X,k)$, so we can get that $h(p,A,r,k)$ is nondecreasing in $p \in [0,1].$

The proof is completed.
\end{proof}
\begin{figure}
  \center
  \includegraphics[width=4in]{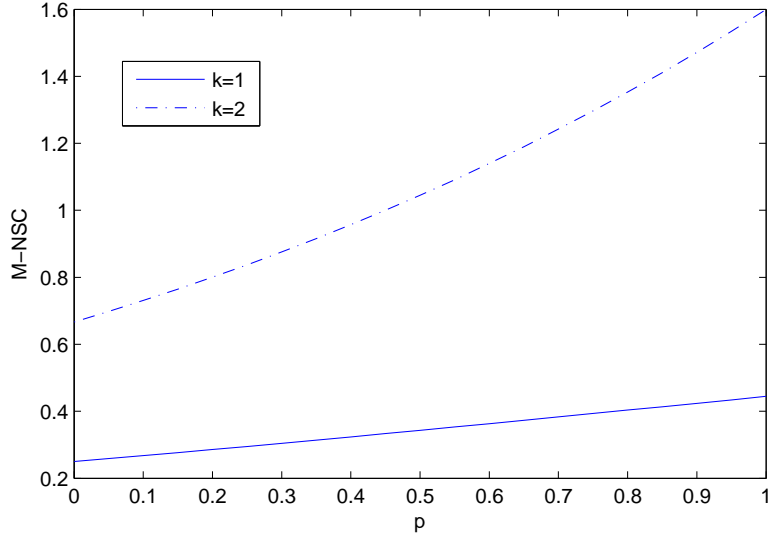}
  \caption{M-NSC in Example 1}
  \label{referencename}
  \end{figure}
\begin{proposition}
The M-NSC as defined in Definition 1, is a continuous function in $p \in [0,1]$.
\end{proposition}
\begin{proof}

As been proved in Proposition 1, $h(p,A,r,k)$ is nondecreasing in $p \in [0,1]$, such that there is jump discontinuous if $h(p,A,r,k)$ is discontinuous at a point. Therefore, it is enough to prove that it is impossible to have jump discontinuous points of $h(p,A,r,k)$.

For convenience, we still use $\theta (p,X,S)$ which is defined in proof of Proposition 1, and the following proof is divided into three steps.

Step 1. To prove that there exist $X\in (N(A))^r$ and a set $S \subset \{1,2\dots n\}$ such that $\theta (p,X,S)=h(p,A,r,k)$.

Let $V=\{X\in ((N(A))^r):\|X_{row \ i}\|_2=1,\ i=1,2\dots n\}$, and it is easy to get that $h(p,A,r,k)=\mathop{\max}\limits_{|S| \leq k} \mathop{\sup}\limits_{X \in V} \theta (p,X,S)$

It needs to be pointing out that, the choice of the set $S\subset \{1,2\dots n\}$ with $|S|\leq k$  is limited, so there exists a set $S^{'}$ with $|S^{'}| \leq k$ such that $h(p,A,r,k)=\mathop{\sup}\limits_{X \in V} \theta (p,X,S^{'}).$

On other hand, $\theta (p,X,S^{'})$ is obviously continuous in $X$ on $V$. Because of the compactness of $V$, there exists $X^{'}\in V$ such that $h(p,A,r,k)=\theta (p,X^{'},S^{'})$.

Step 2. To prove that $\mathop{\lim}\limits_{p \to p_0^{-}} h(p,A,r,k)=h(p_0,A,r,k).$

We assume that $\mathop{\lim}\limits_{p \to p_0^{-}} h(p,A,r,k) \ne h(p_0,A,r,k)$. According to Proposition 1, $h(p,A,r,k)$ is nondecreasing in $p \in [0,1]$, therefore, we can get a sequence of $\{p_n\}$ with $p_n \to p_0^{-}$ such that
\begin{eqnarray}\label{T2}
\mathop{\lim}\limits_{p_n \to p_0^{-}} h(p_n,A,r,k)=M<h(p_0,A,r,k).
\end{eqnarray}

According to the proof in Step 1, there exists $X^{'}\in (N(A))^r$ and $S \subset \{1,2\dots n\}$ such that $h(p_0,A,r,k)=\theta (p_0,X^{'},S^{'})$. It is easy to get that
\begin{eqnarray}
 \mathop{\lim}\limits_{p \to p_0^{-}} \theta(p_n,X,S^{'})=\theta (p,X^{'},S^{'})=h(p_0,A,r,k).
\end{eqnarray}

According to the definition of $\theta (p,X,S)$, it is obvious that
\begin{eqnarray}\label{T3}
h(p_n,A,r,k) \geq \theta(p_n,X^{'},S^{'}),
\end{eqnarray}
however, (\ref{T2}) and (\ref{T3}) contradict each other.

Therefore, we have that $\mathop{\lim}\limits_{p \to p_0^{-}} h(p,A,r,k)=h(p_0,A,r,k)$.

Step 3. To prove that $\mathop{\lim}\limits_{p \to p_0^{+}} h(p,A,r,k)=h(p_0,A,r,k)$, for any $p_0\in [0,1).$

We consider a sequence of $\{p_n\}$ with $p_0 \leq p_n<1$ and $p \to p_0^{+}.$

According to Step 1, there exist $X_n \in V$ and $|S_n|\leq k$ such that
\begin{eqnarray}
h(p_n,A,r,k)=\theta(p_n,X_n,S_n),
\end{eqnarray}
since the choice of $S\subset \{1,2\dots n\}$ with $|S| \leq k$ is limited, there exists two subsequence $\{p_{n_i}\}$ of $\{p_n\}$ , $\{X_{n_i}\}$ of $\{X_n\}$ and a set $S^{'}$ such that
\begin{eqnarray}
\theta(p_{n_i},X_{n_i},S^{'})=h(p_{n_i},A,r,k).
\end{eqnarray}

Furthermore, since $X_n \in V$, it is easy to get a subsequence of $X_{n_i}$ which is convergent. Without of generality,  we assume that $X_{n_i} \to X^{'}$.

Therefore, we can get that $h(p_{n_i},A,r,k)=\theta (p_{n_i},X_{n_i},S^{'}) \to  \theta (p_0,X^{'},S^{'}).$

According to the definition of $h(p_0,A,r,k)$, we can get that $\theta (p_0,X^{'},S^{'}) \leq h(p_0,A,r,k)$, such that $\mathop{\lim}\limits_{p \to p_0^{+}} h(p,A,r,k)=h(p_0,A,r,k)$

Combining Step 2 and Step 3, we show that it is impossible for $h(p,A,r,k)$ to have jump discontinuous.

The proof is completed.
\end{proof}

The concept M-NSC is very important in this paper and it will offer tremendous help in illustrating the performance of $l_{2,0}$-minimization and $l_{2,p}$-minimization, however, M-NSC is difficult to be calculated for large scale matrix. We show the figure of M-NSC in Example 1 in Figure 1.
Combining Proposition 1 and 2, then we can get the first main theorem which shows us the equivalence relationship between $l_{2,0}$-minimization and $l_{2,p}$-minimization.
\begin{theorem}
If every $k$-sparse matrix $X$ can be recovered by $l_{2,0}$-minimization, then there exists a constant $p(A,B,r)$ such that $X$ also can be recovered by $l_{2,p}$-minimization whenever $0<p<p(A,B,r)$.
\end{theorem}
\begin{proof}
According to Proposition 1 and Proposition 2, we can get that $h(0,A,r,k)<1$ if $l_{2,0}$-minimization can recovery every $k$-sparse matrix $X$,

Since $h(p,A,r,k)$ is continuous and nondecreasing at the point $p=0$,  there exists a constant $p(A,B,r)$ and a small enough number $\delta$ that $h(0,A,r,k)<h(p,A,r,k)\leq h(0,A,r,k)+\delta< 1$ for any $p \in (0,p(A,B,r))$.

The proof is completed.
\end{proof}

\section{AN ANALYSIS EXPRESSION OF SUCH P}

In Section 2, we have proved the fact there exists a constant $p(A,B,r)$ such that both $l_{2,p}$-minimization and $l_{2,0}$-minimization have the same solution, however, it is also important to give such an analysis expression of $p(A,B,r)$. In Section 3, we focus on giving an analytic expression of an upper bound of $h(p,A,r,k)$. According to Corollary 2, we can get the equivalence relationship between $l_{2,p}$-minimization and $l_{2,0}$-minimization as long as $h(0,A,r,k)<1$ is satisfied.
In order to reach our goal, we postpone our main theorems and begin with two lemmas.
\begin{lemma}
For any $X\in \mathbb{R}^{n \times r}$, and $\|X\|_{2,p}=\sum_{i=1}^n \|X_{row \ i}\|_2^p$, then we have that $$\|X\|_{2,p} \leq\|X\|_{2,0}^{\frac{1}{p}-\frac{1}{2}}\|X\|_F.$$
\end{lemma}
\begin{proof}
For any $X \in \mathbb{R}^{n \times r}$, without loss of generality, we assume that $\|X_{row\ i}\|_2 =0$, for $i \in \{\|X\|_{2,0}+1,\ldots,n\}$. According to H\"older inequality, we can show that
$$\|X\|_{2,p}^p=\sum_{i=1}^{\|X\|_{2,0}} \|X_{row \ i}\|_2^p \leq \bigg(\sum_{i=1}^{\|X\|_{2,0}} (\|X_{row \ i}\|_2^p)^{\frac{2}{p}}\bigg)^{\frac{p}{2}}\bigg(\sum_{i=1}^{\|X\|_{2,0}} 1\bigg)^{1-\frac{p}{2}}=\|X\|_{2,0}^{1-\frac{p}{2}}\|X\|_F^p.$$
that is $\|X\|_{2,p} \leq\|X\|_{2,0}^{\frac{1}{p}-\frac{1}{2}}\|X\|_F$.
\end{proof}
\begin{lemma}
Give an underdetermined matrix $A \in \mathbb{R}^{m \times n}$. If $h(0,A,r,k)<1$, then we have the following two results.

(a) For any $\|X\|_{2,0} \leq 2k$, we have that
\begin{eqnarray}
 \lambda_{min}^{+}(A)\|X\|_F^2 \leq \|AX\|_F^2 \leq \lambda _{max}(A) \|X\|_F^2.
\end{eqnarray}

(b) For any $X_1,X_2 \in \mathbb{R}^{n \times r}$ with $\|X_i\|_{2,0} \leq k$, $i=1,2$ and $support(X_1) \cap support(X_2)= \varnothing$, we have that
\begin{eqnarray}
 |\langle AX_1,AX_2\rangle|\leq \displaystyle\frac {\lambda_{max}(A)-\lambda_{min}^{+}(A)}{2} \|X_1\|_F\|X_2\|_F.
\end{eqnarray}
\end{lemma}
\begin{proof}

(a) This proof is divided into three steps.

Step 1. To prove that $\|AX\|_F^2\leq \lambda_{max}(A) \|X\|_F^2 $.

For $X\in \mathbb{R}^{n \times r}$, we denote $X=[X_{col\ 1},X_{col\ 2},\dots X_{col\ r}]$, such that
\begin{eqnarray}
 AX=[AX_{col\ 1},AX_{col\ 2},\dots AX_{col\ r}].
\end{eqnarray}

It is obvious that $\|AX_{col\ i}\|_F^2 \leq \lambda_{max}(A)\|X_{col\ i}\|_2^2$, such that $\|AX\|_F^2 \leq \lambda_{max}(A) \|X\|_F^2$

Step 2. To prove that there exists a constant $u>0$ such that $u\|X\|_F \leq \|AX\|_F$, for any $\|X\|_{2,0} \leq k.$

Let the set $V$
\begin{eqnarray}
V=\{u:\|AX\|_F/\|X\|_F \geq u, \ \|X\|_{2,0} \leq 2k\},
\end{eqnarray}
and we assume $\inf\{V\}=0$. i.e., there are a sequence $\{X^{(n)}\}$ with $\|X^{(n)}\|_{2,0}\leq 2k$ such that $\frac{\|AX^{(n)}\|_F}{\|X^{n}\|_F} \to 0$. Without of generality, we assume $\|X^{(n)}\|_F=1$, such that we can get a subsequence $\|X^{(n_t)}\|$ which is convergent, i.e., $X^{(n_t)}\to X^{\ast}$. It is easy to get that $AX^{\ast}=0$ because the function $y(X)=AX$ is a continuous one.

Let $J(X^{\ast})=\{i:\|X_{row \ i}^{\ast}\|_2 \neq 0\}$, since $X^{(n_t)}\to X^{\ast}$, there exists $N_i$ such that $\|X_{row \ i}^{(n_k)}\|\neq 0$ when $t \geq N_i$.

Let $N=\mathop{\max}\limits_{i \in J(X^{\ast})}$, then we can get that $\|X_{row \ i}^{(n_k)}\|_2 \neq 0$ for any $i \in J(X^{\ast}) $, so it is obvious that $\|X^{\ast}\|_{2,0} \leq \|X^{n_k}\|_{2,0} \leq 2k$.

However, this result contradicts the Corollary 1.

Step 3. To prove $\lambda _{min}^{+}(A)\leq u^2$

We assume that $\lambda _{min}^{+}(A)\geq u^2$. Without of generality, we consider $V'=\{X\in R^{n \times r}:\|X\|_F=1\}$. According to the proof in Step 1 and Step 2, there exist a matrix $\tilde{X} \in V'$ such that $\|A\tilde{X}\|_F^2=u\|\tilde{X}\|_F^2$. According to Corollary 1, we can get that $\tilde{X}_{col \ i}\notin N(A)$ for any $i\in [1,2\ldots r]$, since $\|\tilde{X}_{col \ i}\|_0\leq \|\tilde{X}\|_{2,0}\leq2k$.

Furthermore, we can get that $\|A\tilde{X}_{col \ i}\|_2^2=u^2\|\tilde{X}_{col \ i}\|_2^2$, otherwise, we assume that there exists a element $i\in [1,2\ldots r]$ that $\|A\tilde{X}_{col \ i}\|_2^2=\lambda ' \|\tilde{X}_{col \ i}\|_2^2$ with $0<\lambda '\leq \lambda _{min}^{+}(A)$. Considering a matrix $\hat{X}=[\tilde{X}_{col \ i},\tilde{X}_{col \ i} \dots \tilde{X}_{col \ i}]\in R^{n \times r}$, and it is easy to get that  such that$\|\tilde{X}\|_{2,0}\leq 2k$ and $\|A\tilde{X}\|_F=\lambda '\|\tilde{X}\|_F$. The result contradicts the definition of $u$.

Therefore, we can conclude that $\|A_Sx\|_2^2\geq u^2\|x\|_2^2$, for any $x\in \mathbb R^{|S|}$, where $S=support(\tilde{X})$. It is easy to get that the minimum eigenvalues of $A_S^TA_S$ is $u^2$ since $\|A\tilde{X}_{col \ i}\|_2^2=u^2\|\tilde{X}_{col \ i}\|_2^2$, and this result contradicts the definition of $\lambda _{min}^{+}(A)$.

The proof is completed.

(b) According to the definition of inner product of matrices, it is easy to get that
\begin{eqnarray}
\|A(X_1+X_2)\|_F^2=\langle A(X_1+X_2),A(X_1+X_2)\rangle.
\end{eqnarray}
and
\begin{eqnarray}
 \|X_1+X_2\|_F^2=\|X_1-X_2\|_F^2=\|X_1\|_F^2+\|X_2\|_F^2.
\end{eqnarray} since $support(X_1) \cap support(X_2)= \varnothing$.

According to the conclusion (a) in Lemma 2 which has been proved, we have that
\begin{eqnarray}
\notag \frac{|\langle AX_1,AX_2 \rangle|}{\|X_1\|_F\|X_2\|_F}&=&\frac{\left|\|AX_1+AX_2\|_F^2-\|AX_1-AX_2\|_F^2\right|}{4\|X_1\|_F\|X_2\|_F},\\
\notag&\leq& \frac{1}{4}\left(\lambda_{max}(A)\frac{\|X_1+X_2\|_F^2}{\|X_1\|_F\|X_2\|_F}-\lambda_{min}^{+}(A)\frac{\|X_1-X_2\|_F^2}{\|X_1\|_F\|X_2\|_F}\right).\\
&\leq& \displaystyle\frac {\lambda_{max}(A)-\lambda_{min}^{+}(A)}{2}.
\end{eqnarray}
The proof is completed.
\end{proof}

Although Lemma 2 is easy to be proved, it is very important for this paper because it provides us a reason for abandoning the Restricted Isometry Property (RIP) and Restricted Isometry Constant (RIC).

A matrix $A$ is said to have restricted isometry property of order $k$ with restricted isometry constant $\delta_k \in (0,1)$, if $\delta_k$ is the smallest constant such that
\begin{eqnarray}
(1-\delta_k)\|x\|_2 \leq \displaystyle\|Ax\|_2 \leq (1+\delta_k)\|x\|_2,
\end{eqnarray}
for all $k$-sparse vector $x$, where a vector $x$ is said $k$-sparse if $\|x\|_0\leq k.$

In single measurement vector (SMV), RIP and RIC is widely used in many papers and there exists many probabilistic results about RIP. However, the point is to highlight that the existence RIC can guarantee every $k$-sparse solution can be recovered, however it is NP-hard to get RIC for a given matrix $A$ which is satisfied RIP.

The conclusion (a) in Lemma 2 which looks like RIP has an advantage that $\lambda _{max}(A)$ and $\lambda _{min}^{+}(A)$ is easy to be calculated if the matrix $A$ can recovery every $k$-sparse solution.

By contrast, there are many matrices with a particular structure satisfying the condition in Lemma 2, for example, the Vandermonde matrix,
\begin{equation}
A=\left(
  \begin{array}{cccc}
    1 & 1 & ... & 1\\
    t_1 & t_2 & ... & t_n\\
    t_1^2 & t_2^2 & ... & t_n^2\\
    \vdots & \vdots& \ddots & \vdots\\
    t_1^m & t_2^m & ... & t_n^m\\
  \end{array}
\right).
\end{equation}

It is obvious that every sub-matrix $A_S$ is invertible with $|S|\leq k$ as long as $t_i\neq t_j (i,j \in \{1,2..n\} \ and \ i\neq j)$ and $m>2k$, such that every $k$-sparse vector can be recovered.

Now, we present two theorems which are the other main contribution in this paper. Theorem 4 shows us an upper bound of $h(p,A,r,k)$ and Theorem 5 shows us a $p^*(A,B)$ such that $h(p,A,r,k)<1$ when $0<p<p^*(A,B)$.

\begin{theorem}
Given a matrix $A\in \mathbb{R}^{m\times n}$ with $m \leq n$. If $h(0,A,r,k)<1$, then for $p\in (0,1]$, we can get an upper bound of $h(p,A,r,k)$.
\begin{eqnarray}
 \notag h(p,A,r,k)\leq M=\displaystyle\left(\frac{\sqrt{2}+1}{2}\left(\frac{(\lambda-1)(n-2-k)}{2k}+\frac{\lambda-1}{2\sqrt{k}}+\frac{1}{2k}\right)\left({\frac{k}{k+1}}\right)^{\frac{1}{p}}\right)^p,
\end{eqnarray}
where $\lambda =\displaystyle\frac {\lambda_{max}(A)}{\lambda_{min}^{+}(A)}$.

Therefore, we also have that
\begin{eqnarray}
\|X_S\|_{2,p}^p \leq M \|X_{S^C}\|_{2,p}^p,
\end{eqnarray}
for any $X \in (N(A))^r \setminus \{(\textbf{0},\textbf{0}\dots \textbf{0})\}$ and $S\subset \{1,2\dots n\}$ with $|S|\leq k.$
\end{theorem}
\begin{proof}
For any $X \in (N(A))^r \setminus \{(\textbf{0},\textbf{0}\dots \textbf{0})\}$. We define that
\begin{eqnarray}
 x=\displaystyle[\|X_{row \ 1}\|_2,\|X_{row \ 2}\|_2,\dots \|X_{row \ n}\|_2],
\end{eqnarray}
and we consider the index set $S_0$=\{ indices of the largest $k$ values component of $x$\}.

$S_1$=\{ indices of the largest $k+1$ values component of $x$ except $S_0$\}.

$S_2$=\{ indices of the largest $k$ values component of $x$ except $S_0$ and $S_1$\}.

\dots

$S_t$=\{ indices of the rest components of $x$ \}.

According to Corollary 1, we know that $\|x\|_0\geq2k+1$, so both $S_1$ and $S_0$ are not empty and there are only two cases,

(i) $S_0$ and $S_i\ (i=2\ldots t-1)$ all have $k$ elements except $S_t$ possibly.

(ii) $S_0$ has $k$ elements, $S_1$ has less than $k+1$ elements and $S_i\ (i=2\ldots t-1)$ are empty.

Furthermore, in both cases, the set $S_1$ can be divided in two parts.

$S_1^{(1)}$=\{indices of the $k$ largest absolute-values components of $S_1$ \}.

$S_1^{(2)}$=\{indices of the rest components of $S_1$ \}.

It is obvious that $S_1=S_1^{(1)} \cup S_1^{(2)}$ and the set $S_1^{(2)}$ is not empty since $\|x\|_0\geq2k+1$.

According to the definition of $S_1^{(2)}$, it is easy to get that
\begin{eqnarray}
\|X_{S_1^{(2)}}\|_F \leq \displaystyle \sqrt {\frac {1}{k+1}} \|X_{S_1}\|_F,
\end{eqnarray}
 and
\begin{eqnarray}
\|X_{S_1^{(2)}}\|_F \leq \displaystyle  \sqrt {\frac {1}{k}}\|X_{S_0}\|_F,
\end{eqnarray}
such that
\begin{eqnarray}
\|X_{S_1^{(2)}}\|_F^2 \leq\displaystyle {\sqrt {\frac {1}{4k}}} \left(\|X_{S_1}\|_F+\|X_{S_0}\|_F\right)\|X_{S_1^{(2)}}\|_F.
\end{eqnarray}
On the other hand, it needs to be pointed out that
\begin{eqnarray}
 AX=A(X_{S_0}+X_{S_1^{(1)}}+X_{S_1^{(2)}}+X_{S_2}\dots +X_{S_t})=\textbf{0},
\end{eqnarray}
such that
\begin{eqnarray}\label{E1}
\notag\|A(X_{S_0}+X_{S_1^{(1)}})\|_F^2 &=& <A(X_{S_0}+X_{S_1^{(1)}}),-A(X_{S_1^{(2)}}+X_{S_2}\dots +X_{S_t})>\\
 \notag &=& <A(X_{S_0}+X_{S_1^{(1)}}),-AX_{S_1^{(2)}}>\\ &&+ \sum_{i=2}^t(<AX_{S_0},-AX_{S_i}>+<AX_{S_1^{(1)}},-AX_{S_i}>).
\end{eqnarray}
Furthermore, according to Lemma 2, for $i=2\dots t$ and $S_1^{(2)}$, we have that
\begin{eqnarray}\label{E2}
\notag|\langle AX_{S_0},AX_{S_i}\rangle| &\leq& \frac {\lambda_{max}(A)-\lambda_{min}^{+}(A)}{2}\|X_{S_0}\|_F\|X_{S_i}\|_F,\\
|\langle AX_{S_0},AX_{S_1^{(2)}}\rangle| &\leq& \frac {\lambda_{max}(A)-\lambda_{min}^{+}(A)}{2}\|X_{S_0}\|_F\|X_{S_1^{(2)}}\|_F,
\end{eqnarray}
and
\begin{eqnarray}\label{E3}
\notag|\langle AX_{S_1^{(1)}},AX_{S_i}\rangle|&\leq& \frac {\lambda_{max}(A)-\lambda_{min}^{+}(A)}{2}\|X_{S_1^{(1)}}\|_F\|X_{S_i}\|_F,\\
|\langle AX_{S_1},AX_{S_1^{(2)}}\rangle| &\leq& \frac {\lambda_{max}(A)-\lambda_{min}^{+}(A)}{2}\|X_{S_1}\|_F\|X_{S_1^{(2)}}\|_F.
\end{eqnarray}

Substituting the inequalities (\ref{E2}) and (\ref{E3}) into (\ref{E1}), we have
\begin{eqnarray}
\notag \|A(X_{S_0}+X_{S_1^{(1)}})\|_F^2 &\leq& \frac {\lambda_{max}(A)-\lambda_{min}^{+}(A)}{2}\left(\|X_{S_0}\|_F+\|X_{S_1^{(1)}}\|_F\right)\|X_{S_1^{(2)}}\|_F\\&&\notag+\frac {\lambda_{max}(A)-\lambda_{min}^{+}(A)}{2}\left(\sum_{i=2}^t\|X_{S_i}\|_F\right)\left(\|X_{S_0}\|_F+\|X_{S_1^{(1)}}\|_F\right)\\
\notag&\leq&\frac {\lambda_{max}(A)-\lambda_{min}^{+}(A)}{2}\left(\sum_{i=2}^t\|X_{S_i}\|_F+\|X_{S_1^{(2)}}\|_F\right)\left(\|X_{S_0}\|_F+\|X_{S_1^{(1)}}\|_F\right).
\end{eqnarray}

According to Lemma 2, we have that
\begin{eqnarray}
\notag\|X_{S_0}\|_F^2+\|X_{S_1}\|_F^2 &=& \|X_{S_0}+X_{S_1^{(1)}}\|_F^2+\|X_{S_1^{(2)}}\|_F^2 \\
&\leq & \frac{1}{\lambda_{min}^{+}(A)}\|A(X_{S_0}+X_{S_1^{(1)}})\|_F^2+\|X_{S_1^{(2)}}\|_F^2.
\end{eqnarray}

Therefore, we can get that
\begin{eqnarray}\label{E6}
\notag\|X_{S_0}\|_F^2+\|X_{S_1}\|_F^2 &\leq& \left(\|X_{S_0}\|_F+\|X_{S_1}\|_F\right)\Bigg(\frac {\lambda_{max}(A)-\lambda_{min}^{+}(A)}{2\lambda_{min}^{+}(A)}\sum_{i=2}^t\|X_{S_i}\|_F\\
&&\displaystyle+\left(\frac{\lambda_{max}(A)-\lambda_{min}^{+}(A)}{2\lambda_{min}^{+}(A)}
+{\frac{1}{2\sqrt{k}}}\right)\|X_{S_1^{(2)}}\|_F\Bigg)
\end{eqnarray}

According to the definition of $S_i$ ($i\geq 2$), for any $j \in S_i$, it is easy to get that
\begin{eqnarray}
\|X_{row \ j}\|_2^p\leq \frac{1}{k+1}\|X_{S_1}\|_{2,p}^p,
\end{eqnarray}
such that
\begin{eqnarray}
\|X_{row \ j}\|_2^2\leq (k+1)^{-\frac{2}{p}}\|X_{S_1}\|_{2,p}^2,
\end{eqnarray}
Since the set $S_i$ has $k$ elements, it is easy that
\begin{eqnarray}\label{E4}
\|X_{S_i}\|_F\leq \sqrt{k}(k+1)^{-\frac{1}{p}}\|X_{S_1}\|_{2,p},
\end{eqnarray}
 and it is also obvious that
\begin{eqnarray}\label{E5}
 \|X_{S_1^{(2)}}\|_F\leq (k+1)^{-\frac{1}{p}}\|X_{S_1}\|_{2,p}.
\end{eqnarray}
Substituting these inequalities (\ref{E4}) and (\ref{E5}) into (\ref{E6}), we can get that
\begin{eqnarray}\label{E7}
\displaystyle\notag
\|X_{S_0}\|_F^2+\|X_{S_1}\|_F^2 & \leq & (\|X_{S_0}\|_F+\|X_{S_1}\|_F)\\
&&\left(\frac{1}{2\sqrt{k}}+\frac{\lambda-1}{2}+\frac{(t-1)(\lambda-1)\sqrt{k}}{2}\right)(k+1)^{-\frac{1}{p}}\|X_{S_1}\|_{2,p},
\end{eqnarray}
where $\lambda=\frac{\lambda_{max}(A)}{\lambda_{min}^{+}(A)}$.

We denote that
\begin{eqnarray}
C=\left(\frac{1}{2\sqrt{k}}+\frac{\lambda-1}{2}+\frac{(t-1)(\lambda-1)\sqrt{k}}{2}\right)(k+1)^{-\frac{1}{p}}\|X_{S_1}\|_{2,p},
\end{eqnarray}
then we can rewrite (\ref{E7}) that
\begin{eqnarray}
\|X_{S_0}\|_F^2+\|X_{S_1}\|_F^2 \leq C(\|X_{S_0}\|_F+\|X_{S_1}\|_F).
\end{eqnarray}

It is easy to get that
\begin{eqnarray}
(\|X_{S_0}\|_F-C)^2+(\|X_{S_1}\|_F-C)^2\leq \frac{C^2}{2},
\end{eqnarray}
such that
\begin{eqnarray}
\|X_{S_0}\|_F \leq \frac{\sqrt{2}+1}{2}C.
\end{eqnarray}

According to Lemma 1, we have that
\begin{eqnarray}
\|X_{S_0}\|_{2,p}\leq k^{\frac{1}{p}-\frac{1}{2}}\|X_{S_0}\|_F\leq k^{\frac{1}{p}-\frac{1}{2}}\frac{\sqrt{2}+1}{2}C.
\end{eqnarray}

We notice that
\begin{eqnarray}
tk+1\leq n \leq (t+1)k+1,
\end{eqnarray}
such that $t\leq \frac {n-2}{k}$. Therefore, we have that
\begin{eqnarray}
\|X_{S_0}\|_{2,p}^p \leq M \|X_{S_0^C}\|_{2,p}^p,
\end{eqnarray}
where $ M=\displaystyle\left(\frac{\sqrt{2}+1}{2}\left(\frac{(\lambda-1)(n-2-k)}{2k}+\frac{\lambda-1}{2\sqrt{k}}+\frac{1}{2k}\right)\left({\frac{k}{k+1}}\right)^{\frac{1}{p}}\right)^p$.

According to the definition of $S_0$, $S_0$ contains the $k$ largest element in $x$, so it is obvious that the inequality $\|X_{S}\|_{2,p}^p \leq M \|X_{S^C}\|_{2,p}^p$ holds for any $S\subset \{1,2\dots n\}$ with $|S|\leq k$ and $h(p,A,r,k)\leq M$.

The proof is completed.
\end{proof}

\begin{theorem}
Given an underdetermined matrix $A\in R^{m\times n}$ with $m \leq n$. and denote $S^*=|support(A^T(AA^T)^{-1}B)|$. If every $k$-sparse matrix $X^{\ast} \in \mathbb{R}^{n \times r}$ can be recovered by $l_{2,0}$-minimization, then $X^{\ast}$ can also be recovered by $l_{2,p}$-minimization, for any $0<p< p^{\ast}(A,B)$, where

\begin{eqnarray}
\displaystyle p^*(A,B)=max\left\{f\left(S^*\right),f\left(\left\lceil \frac{m}{2} \right\rceil\right),f\left(\left\lceil \frac{n-2.5}{2} \right\rceil+1\right)\right\},
\end{eqnarray}
with
\begin{eqnarray}
\displaystyle f(x)=\frac{\ln (x+1)- \ln x}{\ln (\sqrt{2}+1)+\ln (\lambda n-n-2\lambda+3)-\ln 4}
\end{eqnarray}
and $\lambda=\displaystyle\frac{\lambda_{max}(A)}{\lambda_{min^{+}}(A)}.$

\end{theorem}
\begin{proof}
According to Theorem 4, we can get the equivalence between $l_{2,0}$-minimization and $l_{2,p}$-minimization, as soon as $M<1$ where $M$ is defined in Theorem 4.
\begin{eqnarray}
 \displaystyle\frac{\sqrt{2}+1}{2}\left(\frac{(\lambda-1)(n-2-k)}{2k}+\frac{\lambda-1}{2\sqrt{k}}+\frac{1}{2k}\right)\left(\frac{k}{k+1}\right)^{\frac{1}{p}}<1
\end{eqnarray}
Due to the integer-value virtue of $\|X\|_{2,0}$, we have that,
$$\frac{\sqrt{2}+1}{2}\left(\frac{(\lambda-1)(n-2-k)}{2k}+\frac{\lambda-1}{2\sqrt{k}}+\frac{1}{2k}\right)\leq \frac{(\sqrt{2}+1)\left[(\lambda-1)(n-3)+\lambda\right]}{4}$$
Therefore, we can get a range of such $p$ from the following inequality,
$$\displaystyle \left(\frac{(\sqrt{2}+1)(\lambda n-n-2\lambda+3)}{4}\right)\left(\frac{k}{k+1}\right)^{\frac{1}{p}}<1$$
It is easy to solve this inequality, and we can get that
\begin{eqnarray}
\displaystyle p<f(k)=\frac{\ln(k+1)-\ln k}{\ln(\sqrt{2}+1)+\ln (\lambda n-n-2\lambda+3)-\ln 4}
\end{eqnarray}

Furthermore, it is easy to find that $f(k)$ is nonincreasing in $k$, and according to Corollary 1, we have that $k\leq S^*=|support(A^T(AA^T)^{-1}B)|$, $k\leq \left\lceil \frac{m}{2} \right\rceil$, and $k\leq \left\lceil \frac{n-2.5}{2} \right\rceil+1$
Therefore, it is obvious that $h(p,A,r,k)<1$ when $0<p<p^*(A,B)=max\left\{f\left(S^*\right),f\left(\left\lceil \frac{m}{2} \right\rceil\right),f\left(\left\lceil \frac{n-2.5}{2} \right\rceil+1\right)\right\}$.

The proof is completed.
\end{proof}

Now, we present one example to demonstrate the validation of our main contribution in this paper.
\begin{example}
We consider an underdetermined system $AX=B$, where\\
$A=\left(
\begin{array}{ccccc}
1.3746 & -1.2656 & -0.3614 & 1.2431 & 1.8634 \\
1.3495 & -1.4414 & -0.5365 & -0.4636 & -1.7457 \\
-1.2791 & -1.7630 & 0.7327 & 1.6626 & -0.9738 \\
1.4309 & 1.5486 & -0.0205 & 1.9911 & -1.0973
\end{array}
\right)$ and $B=\left(
\begin{array}{cc}
-3.1290 & -4.9924\\
0.3043 & 2.0500\\
-0.7892 & 0.1846\\
2.6459 & 3.7432
\end{array}
\right)$
\end{example}

It is easy to verify the unique sparse solution to $l_{2,0}$-minimization is
\begin{eqnarray}
X^{\ast}=\left(
\begin{array}{cc}
0 & 0 \\
1 & 1 \\
0 & 0 \\
0 & 0 \\
1 & 2
\end{array}
\right)
\end{eqnarray}
and $N(A)=\{ax:a \in R,x=[0.3217,-0.0331,0.9291,-0.1754,-0.0371]^T\}$.

So the solution to $AX=B$ can be expressed as the following form:
\begin{eqnarray}
X=\left(
\begin{array}{cc}
0.3217m & 0.3217n \\
1-0.0331m & 1-0.0331n \\
0.9291m & 0.9291n \\
-0.1754m & -0.1754n \\
1-0.0371m & 2-0.0371n
\end{array}
\right)
\end{eqnarray}
Where $m,n \in \mathbb R$, such that
\begin{eqnarray}
\notag\|X\|_{2,p}^p&=&(0.3217^2+0.9291^2+0.1754^2)^{\frac {p}{2}}(m^2+n^2)^{\frac {p}{2}}\\
\notag &&+[(1-0.0331m)^2+(1-0.0331n)^2]^{\frac {p}{2}}\\
&&+[(1-0.0371m)^2+(1-0.0371n)^2]^{\frac {p}{2}}
\end{eqnarray}
\begin{figure}[htbp]
\centering
\subfigure{
\begin{minipage}{7cm}
\centering
\includegraphics[width=\textwidth]{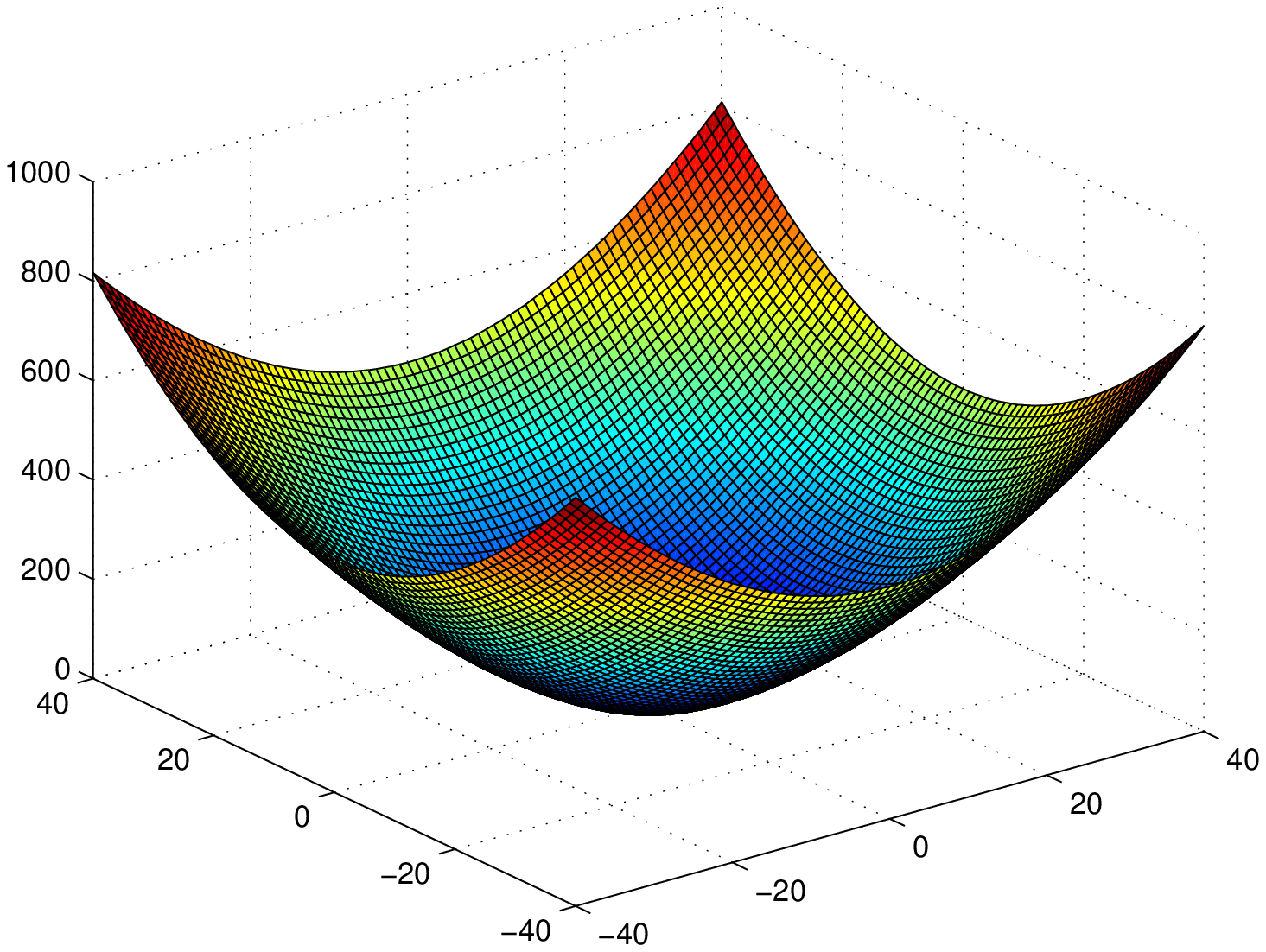}

\end{minipage}
}
\subfigure{
\begin{minipage}{7cm}
\centering
\includegraphics[width=\textwidth]{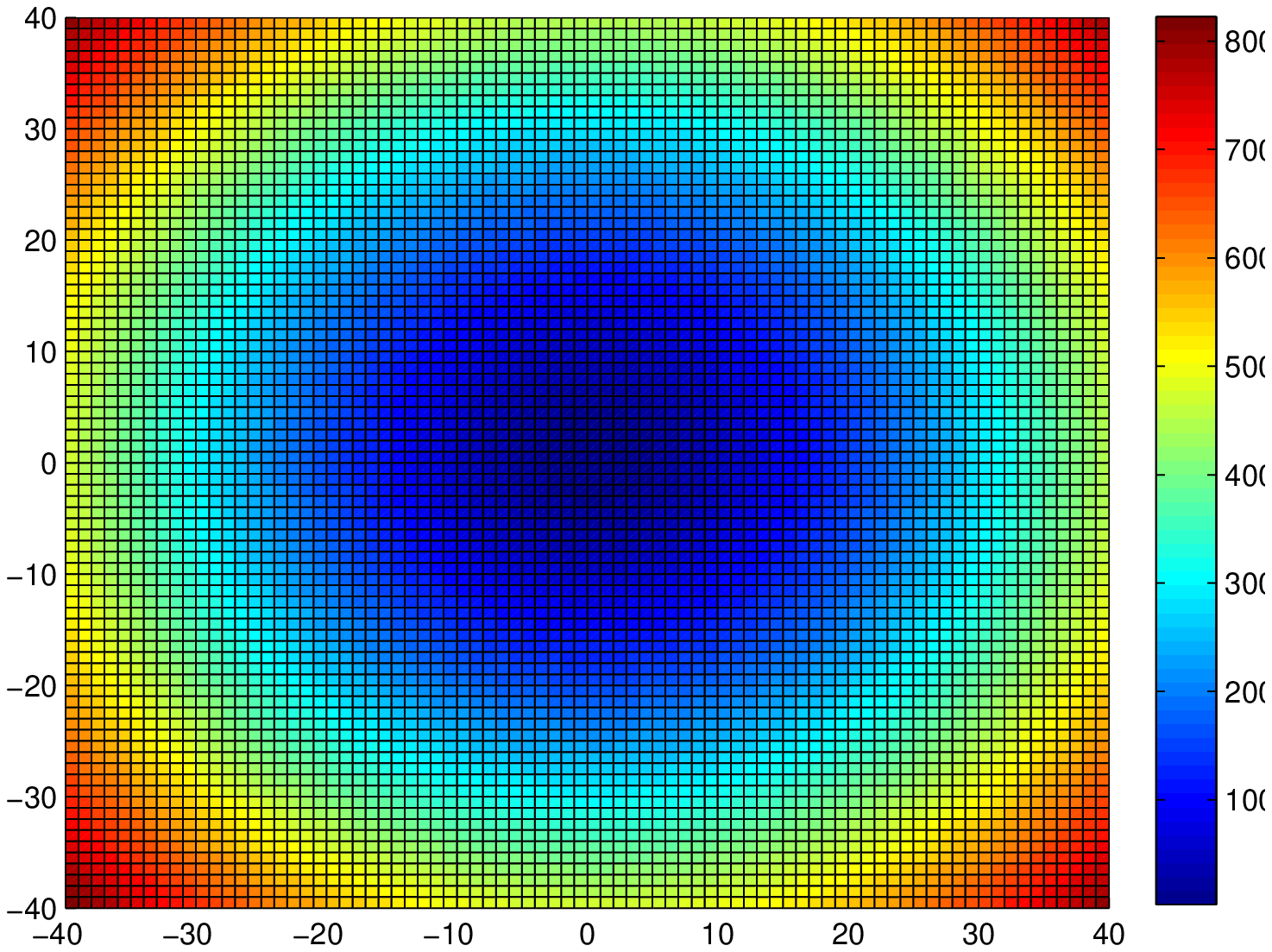}

\end{minipage}
}
\caption{p=0.8175}
\end{figure}
\begin{figure}[htbp]
\centering
\subfigure{
\begin{minipage}{7cm}
\centering
\includegraphics[width=\textwidth]{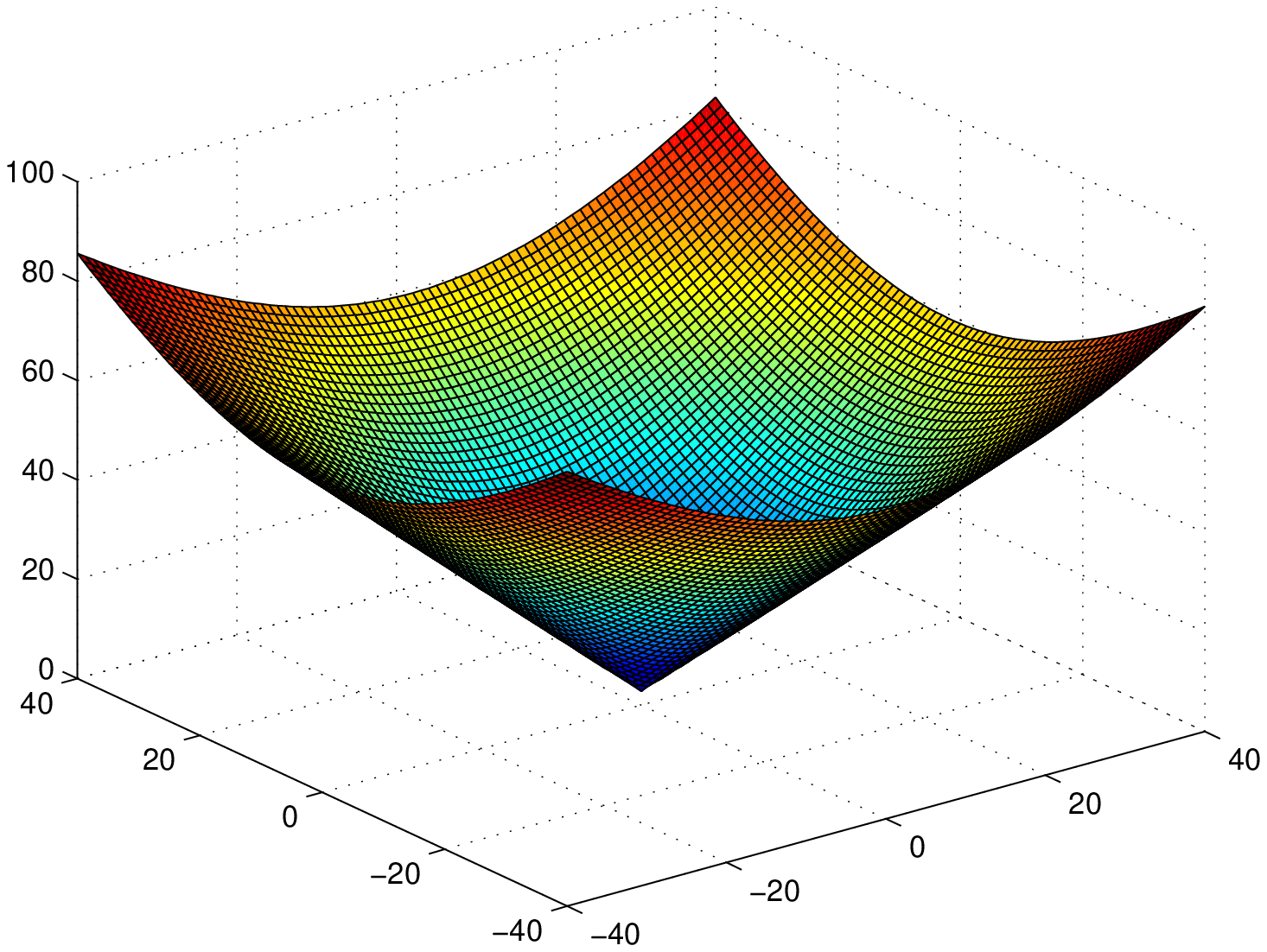}

\end{minipage}
}
\subfigure{
\begin{minipage}{7cm}
\centering
\includegraphics[width=\textwidth]{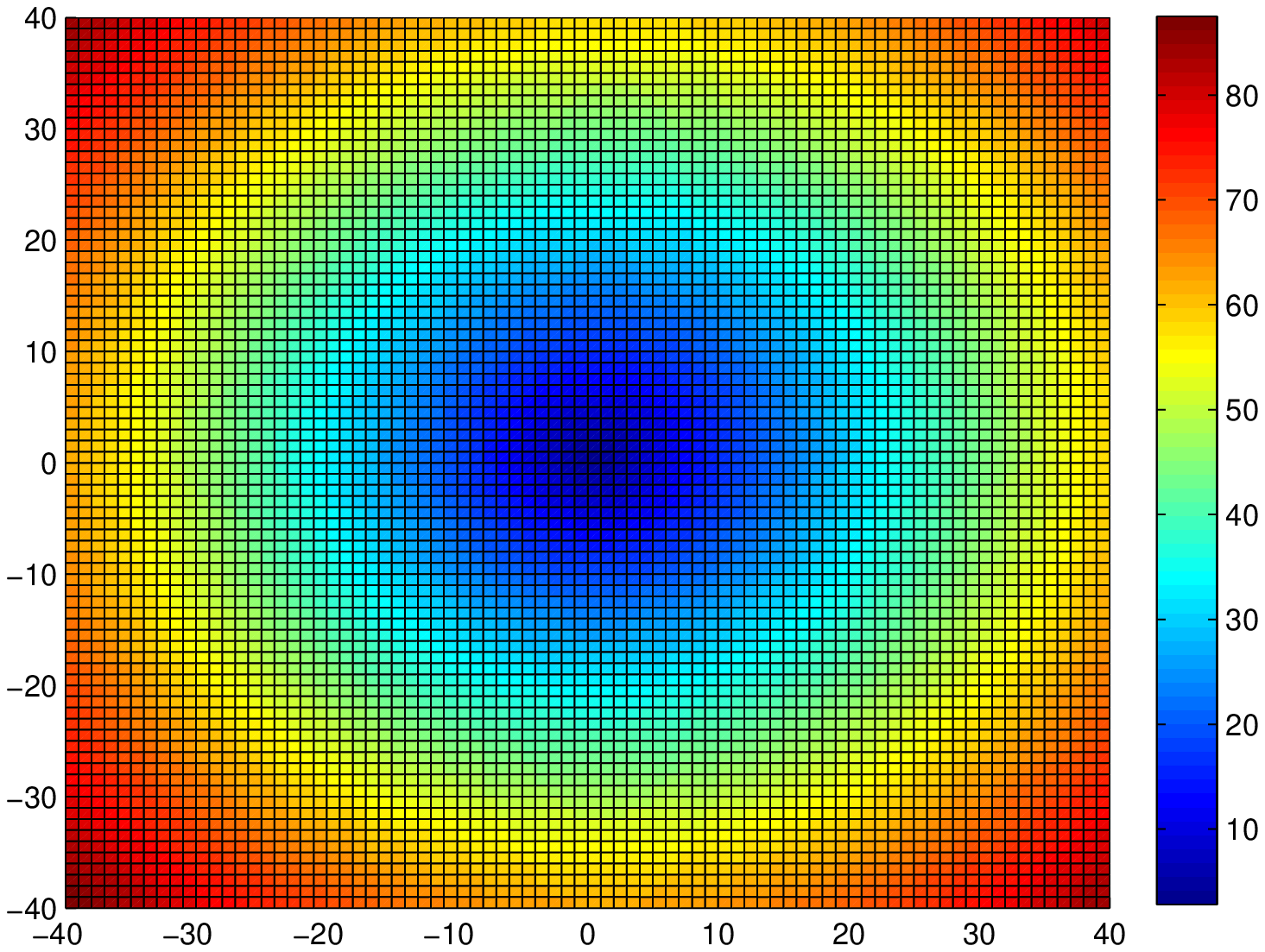}

\end{minipage}
}
\caption{p=0.5}
\end{figure}
\begin{figure}[htbp]
\centering
\subfigure{
\begin{minipage}{7cm}%²¢ÅÅ·ÅÁ½ÕÅͼƬ£¬Ã¿ÕÅÕ¼Ò³ÃæµÄ0.5£¬ÏÂͬ¡£
\centering
\includegraphics[width=\textwidth]{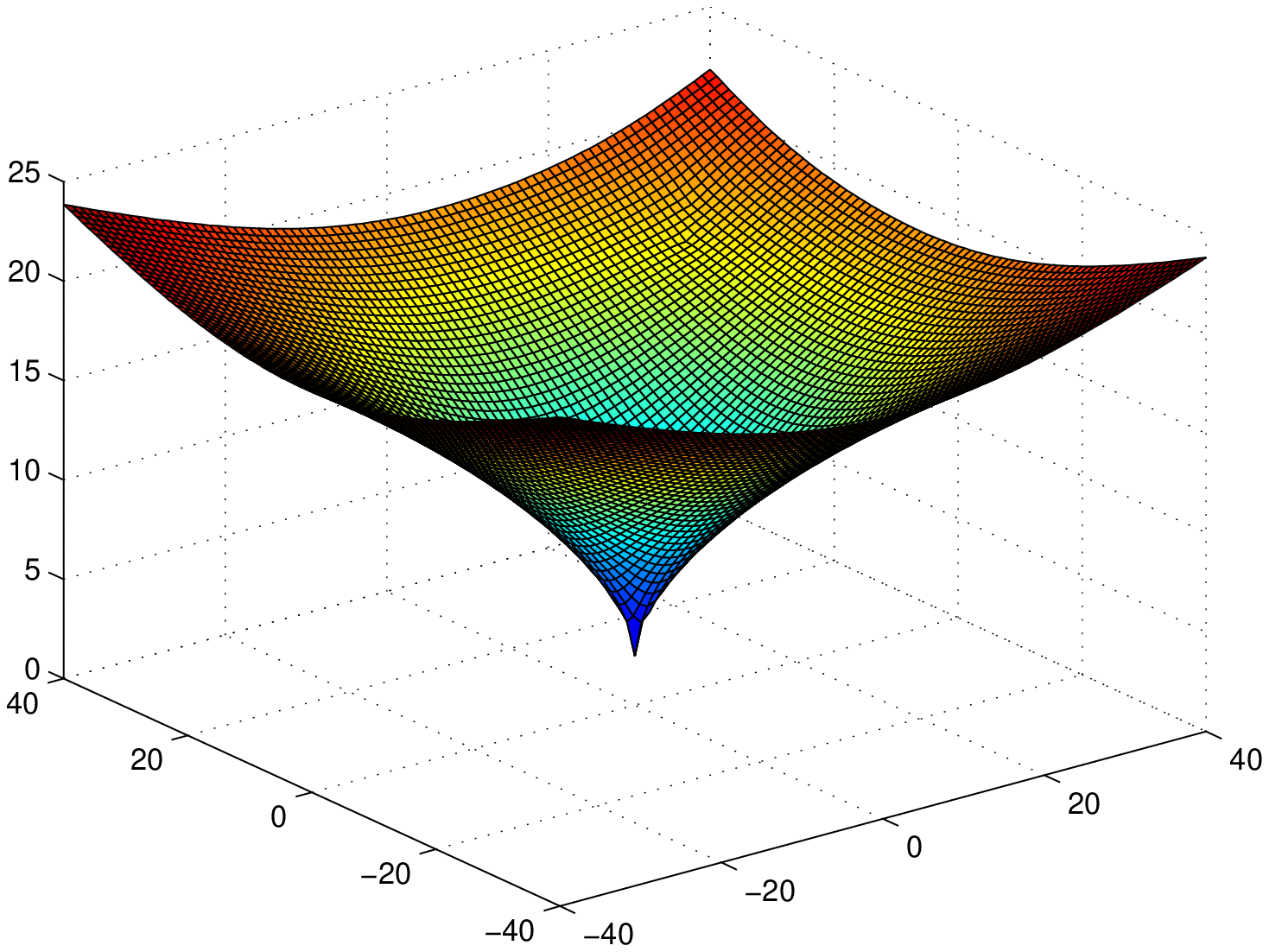}

\end{minipage}
}
\subfigure{
\begin{minipage}{7cm}
\centering
\includegraphics[width=\textwidth]{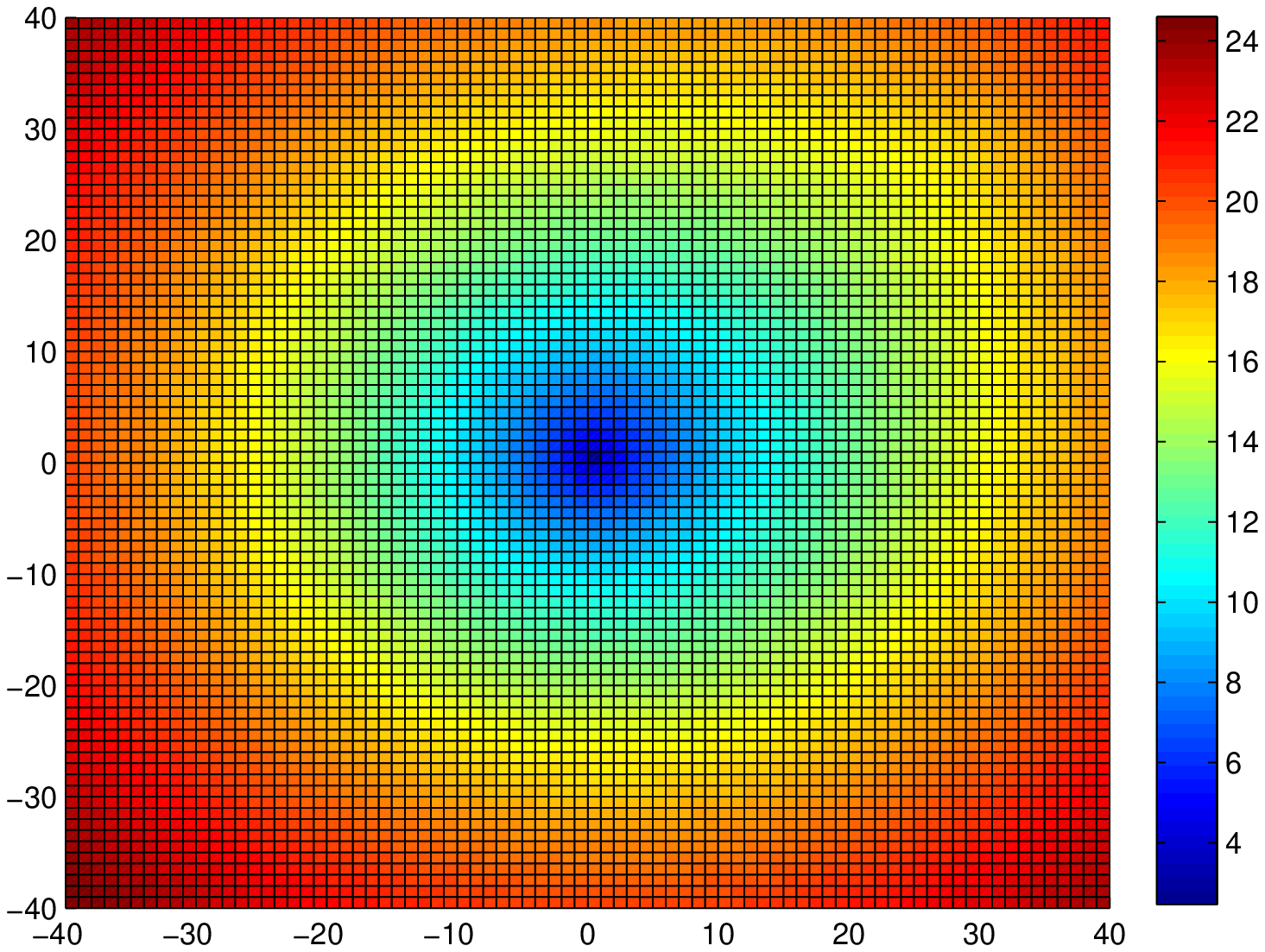}

\end{minipage}
}
\caption{p=0.3}
\end{figure}
Then we will verify the result in Theorem 5, it is easy to get that $p^{*}=0.8176$, and we show the cases when $p=0.8175$, $p=0.5$ and $p=0.3$. in Figure 2, Figure 3 and Figure 4.

It is obvious that $\|X\|_{2,p}$ has the minimum point at $n=m=0$ which is the original solution to $l_{2,0}$-minimization.
\section{CONCLUSION}
In this paper we have studied the equivalence relationship between $l_{2,0}$-minimization and $l_{2,p}$-minimization, and we give an analysis expression of such $p^{\ast}(A,B)$.

Furthermore, it needs to be pointed out that the conclusion in Theorem 4 and Theorem 5 is valid in single measurement vector problem. i.e. $l_p$-minimization also can recovery the original unique solution to $l_0$-minimization when $0<p<p^{\ast}$.

However, the analysis expression of such $p^{\ast}$ in Theorem 5 may not be the optimal result. In this paper, we consider all the underdetermined matrix $A\in \mathbb{R}^{m \times n}$ and $B\in \mathbb{R}^{m \times r}$ from a theoretical point of view. So the result can be improved with a particular structure of the matrix $A$ and $B$. For example, the underdetermined $A$ have restricted isometry property (RIP) of order $k$ which is widely used in many algorithms.
The authors think the answer to this problem will be an important improvement for the application of $l_{2,p}$-minimization. In conclusion, the authors hope that in publishing this paper, a brick will be thrown out and be replaced with a gem.

\bibliographystyle{plain}
\bibliography{myreference}
\end{document}